\documentclass[12pt]{article}

\usepackage{natbib}
\usepackage{amsmath,amsfonts,amsthm}
\usepackage{dsfont} % indicator function
\usepackage{enumerate}
\usepackage[titletoc,title]{appendix}

\newtheorem{theorem}{Theorem}
\newtheorem{definition}{Definition}

\newtheorem{lemma}{Lemma}

\newtheorem{example}{Example}

\usepackage[pdftex]{graphicx}
\usepackage{caption}
\usepackage{subcaption}

\newcommand*{\Scale}[2][4]{\scalebox{#1}{$#2$}}
\graphicspath{{figures/}}

\newcommand{\blind}{1}

\begin{document}

\def\spacingset#1{\renewcommand{\baselinestretch}%
{#1}\small\normalsize} \spacingset{1}

%%%%% Title

\if1\blind
{
  \title{\bf Conformal prediction intervals for the individual treatment effect}
  \author{Danijel Kivaranovic\textsuperscript{a} \quad Robin Ristl\textsuperscript{b} \quad Martin Posch\textsuperscript{b} \quad Hannes Leeb\textsuperscript{a,c} \\ \\
 \textsuperscript{a}Department of Statistics and Operations Research, University of Vienna  \\
 \textsuperscript{b}Center for Medical Statistics, Informatics and Intelligent Systems,\\Medical University of Vienna \\
 \textsuperscript{c}Data Science @ Uni Vienna, University of Vienna
  }
  \date{}
  \maketitle
} \fi

\if0\blind
{
  \bigskip
  \bigskip
  \bigskip
    \title{\bf Conformal prediction intervals for the individual treatment effect}
  \date{}
  \maketitle
  \medskip
} \fi

\date{}

%%%%% Abstract

\bigskip
\begin{abstract}
We propose several prediction intervals procedures for the individual treatment effect with either finite-sample or asymptotic coverage guarantee in a non-parametric regression setting, where non-linear regression functions, heteroskedasticity and non-Gaussianity are allowed. The construct the prediction intervals we use the conformal method of \cite{vovk2005}. In extensive simulations, we compare the  coverage probability and interval length of our prediction interval procedures. We demonstrate that complex learning algorithms, such as neural networks, can lead to narrower prediction intervals than simple algorithms, such as linear regression, if the sample size is large enough. 
\end{abstract}

\noindent%
\vfill

\newpage
\spacingset{1.45} % DON'T change the spacing!

\section{Introduction}
In usual randomized controlled clinical trials, the average treatment effect, averaged over the population under study, is estimated and used as predictor for the expected treatment effect in a new patient.
However, treatment effects may depend on individual patient characteristics and the personalized medicine paradigm aims to tailor treatment decisions to these characteristics \citep{pellegrini2019proof,graf2020optimized}. When comparing an experimental treatment to a control treatment, the relevant quantity to establish a personalized treatment decision is the so called \textit{individual treatment effect}, defined as the difference between the patient's potential outcomes under treatment and under control \citep{rubin1974estimating,rubin2005causal}. In this paper, we propose several prediction intervals for the individual treatment effect with either finite-sample or asymptotic coverage guarantee. Our intervals are valid in a non-parametric regression setting, where non-linear regression functions, heteroskedasticity and non-Gaussianity are allowed.
%In contrast, the prediction intervals of the existing literature typically require a true linear model and homoscedastic Gaussian errors (ADD: Referenzen).
The construction of our prediction intervals relies on the conformal method of \cite{vovk2005,shafer2008}.

%Across a patient population, the variability in individual treatment effects may be decomposed into between-patient variability due to observed patient characteristics, between-patient variability due unobserved characteristics and within-patient variability that reflects the assumption that the same patient may respond differently to a treatment at different occassions \citep{Senn}. 

To predict individual treatment effects based on a set of co-variates, regression models for, both, the response under treatment and under control can be fitted. This can be achieved, for example,  if interaction terms for the treatment group indicator and the considered predictive variables are included or if separate regression models  for each treatment group are fitted. Alternatively, more complex non-parametric or machine-learning models can be applied to model non-linear associations \citep{lamont2018identification}.
The expected individual treatment effect can then be estimated as difference of predictions under treatment and under control, given a patient's co-variate values. 

In order to reliably apply this concept for decision making, the uncertainty attached to an estimated individual treatment effect needs to quantified. \cite{ballarini2018subgroup} studied the application of confidence intervals for expected individual treatment effects to select subgroups with improved treatment benefit. %\cite{huang2012assessing} proposed a framework to calculate sensitivity and specificity in binary classifications of an individual treatment benefit and provide confidence intervals for these quantities.
While this approach covers uncertainty due to the prediction model being estimated from a finite training sample, a complete assessment of uncertainty in an individual prediction needs to take into account the residual distribution of the effect, which comprises unexplained between-patient variability and within-patient variability.  (The latter reflects the possibility that the same patient may respond differently to a treatment at different occasions; but only in rare settings which allow for repeated observations of the effect in the same patient these sources of variability may be identified separately \citep{senn2016mastering}.)

A suitable way to communicate the extent of remaining uncertainty is in terms of prediction intervals. %which are defined through quantiles of the residual distribution.
The calculation of individual prediction intervals from models for the individual treatment effect is complicated by the fact that usually each patient in the training data set is observed under one of the two treatment conditions only. Hence, residuals of the individual treatment effect cannot be observed.
To overcome this problem, in this manuscript we propose a method to calculate 
%exact
prediction intervals for the individual treatment effect by an extension of the conformal inference framework towards the difference in hypothetical outcomes under two treatment conditions. Conformal inference allows to estimate intervals in the residual distribution of an arbitrary prediction model without distributional assumptions and take into account uncertainty in the model estimation procedure. The resulting intervals are
exact
in the sense that the average coverage probability across the studied population is controlled.

In Section \ref{sec_theorems} we introduce the notation and present our main results: Theorem \ref{th_te1} and \ref{th_te1_b} provide prediction intervals with finite-sample coverage, and Theorem \ref{th_te2} provides prediction intervals with asymptotic coverage. Theorem \ref{th_te1} and \ref{th_te1_b} have weaker assumptions but are therefore more conservative. Theorem \ref{th_te2} has stronger assumptions, namely, consistency of the prediction procedure and Gaussian error distribution but provides considerably shorter intervals. All three theorems assume the existence of a prediction intervals procedure for the outcome conditional on the assigned treatment. In Section \ref{sec_conf_pi}, we show how to construct such a prediction interval procedure using the conformal method. In Section \ref{sec_sim} we compare the intervals in several simulation settings. The paper ends with a discussion in Section \ref{sec_discussion}. The proofs of all results are given in the appendix.

\section{Prediction intervals for the individual treatment effect} \label{sec_theorems}

\subsection{Notation and definitions}

We denote by $X$ the random covariate vector with values in $\mathbb R^d$, $d \in \mathbb N$. Let $T$ be the $\{-1,1\}$-valued random variable that indicates whether a patient was assigned to the treatment group ($T=1$) or control group ($T=-1$). We assume that the random variables $X$ and $T$ are independent. This means, patients are assigned to treatment or control group independent of their condition, which is described by $X$. (We note that this setting excludes stratified sampling, because it would introduce a dependence structure between $X$ and $T$.) Let $Y$ denote the real-valued outcome of interest. We can decompose $Y$ into
\begin{equation*}
  Y = f(X,T) + e_{X,T},
\end{equation*}
where $f(X,T) = \mathbb{E}(Y|X,T)$ is called the regression function and $e_{X,T}=Y-\mathbb{E}(Y|X,T)$ the error. The $X$ and $T$ in the subscript of $e_{X,T}$, denote the dependence of $e_{X,T}$ on $X$ and $T$. Each patient is only assigned to one of the two treatment groups, which means that the random variable
\begin{equation*}
  Y' = f(X,-T) + e_{X,-T}
\end{equation*}
is not observable. Note that $Y'$ denotes the outcome of the patient assigned to the opposite treatment group. The individual treatment effect $\tau(X)$ is defined as $T(Y-Y')$ or, equivalently,
\begin{equation*}
  \tau(X) = f(X,1) - f(X,-1) + e_{X,1} - e_{X,-1}.
\end{equation*}
Clearly, because $Y'$ is not observable, $\tau(X)$ is also not observable. We note that $\tau(X)$ is not to be confused with $f(X,1) - f(X,-1)$ which is also sometimes referred to as (predicted) individual treatment effect in the literature. Conditional on $X$, $f(X,1) - f(X,-1)$ is deterministic, while $\tau(X)$ is still random because of the error $e_{X,1} - e_{X,-1}$. This also means that, in a well-defined setting, the length of confidence intervals for $f(X,1) - f(X,-1)$ converge to 0 as sample size grows while the length of prediction intervals for $\tau(X)$ cannot converge to 0 because the variance of $e_{X,1} - e_{X,-1}$ is not explainable. In Theorem $\ref{th_te1}$ we make no assumption about the distribution of $e_{X,1}$ and $e_{X,-1}$, while Theorem \ref{th_te1_b} requires conditional independence and Theorem \ref{th_te2} conditional positive correlation given $X$, respectively. We note that conditional independence between $e_{X,1}$ and $e_{X,-1}$ given $X$ is often considered as worst-case scenario, because unobserved explanatory variables typically cause a positive correlation structure between the errors.

%We note that no assumption on the dependence of $e_{X,1}$ and $e_{X,-1}$ is made so far. Typically, $e_{X,1}$ and $e_{X,-1}$ are positively correlated because unobserved explanatory variables cause a positive dependence structure between the errors. 

Let $D_n = ((X_i,Y_i,T_i))_{1 \leq i \leq n}$ denote a dataset of $n \in \mathbb N$ observations, where the $(X_i,T_i,Y_i)$s are i.i.d. copies of $(X,Y,T)$. We only consider the non-degenerate case where $0<\mathbb P \left(T=1 \right) < 1$. We denote by $\mathbb D^n = (\mathbb R^d \times \mathbb R \times \{-1,1\})^n$ the range of $D_n$. Throughout this paper, $(X,Y,T)$ denotes a new observation that is independent of $D_n$. Let $\alpha \in (0,1)$ be the confidence level. We cannot define a prediction interval procedure for the individual treatment effect $\tau(X)$ based on the unobservable vector $(\tau(X_1), \dots, \tau(X_n))'$. Therefore, we propose to construct a prediction interval $[l_{D_n}(X,-1), u_{D_n}(X,-1)]$ for $Y$ conditional on the event $\{T=-1\}$ and another prediction interval $[l_{D_n}(X,1), u_{D_n}(X,1)]$ for $Y$ conditional on the event $\{T=1\}$. Then, we combine these two prediction intervals to obtain a prediction interval for $\tau(X)$. In Section \ref{sec_conf_pi}, we show how to modify the conformal method (\cite{vovk2005,shafer2008}) to construct prediction intervals for $Y$ with coverage probability $1-\alpha$ conditional on $\{T=-1\}$ and $\{T=1\}$, respectively. The construction of these intervals merely requires i.i.d.\ data and no further assumption on the distribution of the data or the learning algorithm.

The accuracy of the proposed procedure (in terms of coverage probability and interval length), depends on how accurately we can estimate the regression function $f(x,t)$ and how accurately we can estimate the conditional distribution of $e_{X,T}$ conditional on $X$ and $T$. 
\begin{definition} \label{def_pred}
  A prediction procedure $\mathcal A_n$ is a mapping from $\mathbb D^n$ into the set of all measurable functions $\{f:\mathbb R^d \times \{-1,1\} \to \mathbb R\}$. We set
  \begin{equation*}
    f_{D_n} = \mathcal A_n(D_n).
  \end{equation*}
  We say that $\mathcal A_n$ is a consistent prediction procedure if for any compact set $K \subseteq \mathbb R^d$ and any $\epsilon > 0$
  \begin{equation*}
    \mathbb P \left( \sup_{(x,t) \in K \times \{-1,1\}} |f_{D_n}(x,t) - f(x,t)| < \epsilon \right) \to 1 \quad \text{as } n \to \infty.
  \end{equation*}
\end{definition}
Typically, the efficacy of the treatment $t$ interacts with the covariate vector $x$. In the following, we give examples of prediction procedures $\mathcal A_n$ that model these interactions. Furthermore, we discuss when these procedures are consistent in the sense of Definition \ref{def_pred}:
\begin{enumerate}[(i)]
    \item The least squares method is the most common and most frequently used prediction procedure. The function $f_{D_n}$ is of the form 
    \begin{equation*}
      f_{D_n}(x,t) = \hat \beta_{0} + \hat \beta_{1} t + \sum_{i=1}^d \hat \gamma_{i} x_i + t \sum_{i=1}^d \hat \delta_{i} x_i,
    \end{equation*}
    where the $\hat \beta$s, $\hat \gamma$s and $\hat \delta$s are chosen by minimizing the mean squared error on the dataset $D_n$. Interactions between $x$ and $t$ are explicitly modeled by the last sum. The linear least squares method is only consistent if the regression function $f(x,t)$ is a linear function of $x$ and $t$ and the components of $x$ only interact with $t$.
    \item Kernel regression is a non-parametric prediction procedure. The function $f_{D_n}$ is defined as
    \begin{equation*}
      f_{D_n}(x,t) = \sum_{i=1}^n Y_i \hat w_{i}(x,t),
    \end{equation*}
    where the $\hat w_{i}$s are $D_n$-depended weight functions such that $\hat w_{i}(x,t)$ is large if $(x,t)$ is close to $(X_i,T_i)$ and small if $(x,t)$ is far from $(X_i,T_i)$. Interaction between $x$ and $t$ are implicitly modeled by the data-dependent weight functions. Sufficient conditions for consistency of kernel regression methods are given in \cite{stone1977}.
    \item The neural network is a prediction procedure that has recently gained a lot of attention because of its predictive performance. A standard neural network $f_{D_n}$ with one hidden layer, with $k \in \mathbb N$ nodes and non-linear activation function $\sigma: \mathbb R \to \mathbb R$ is defined as 
    \begin{align*}
    f_{D_n}(x,t) = \hat W_2 ~ \sigma(\hat W_1 (x',t)' + \hat b_1) + \hat b_2,
    \end{align*} 
    where $\hat W_{1}$ is a $D_n$-dependent $k \times (d+1)$ weight matrix;  $\hat b_{1}$ is a $D_n$-dependent $k$-dimensional bias vector; $\sigma$ is applied component-wise; $\hat W_{2}$ is a $D_n$-dependent $1 \times k$ weight matrix and $\hat b_2$ is a $D_n$-dependent $1$-dimensional bias vector. A neural network is called deep, if the number of hidden layers and hidden nodes is large. The larger the network, the more interactions between $x$ and $t$ can be modeled. Neural networks are widely known as universal approximators \citep{cybenko1989,hornik1989} and, more recently, consistency properties from statistical perspective have been analyzed by \cite{bauer2019}.
\end{enumerate}
Note that if the sample size $n$ is small, some prediction procedures are not well defined (e.g. if the number of parameters in a linear regression model is larger than $n$). In such cases, $f_{D_n}$ is simply to be interpreted as the $0$ function.

\begin{definition} \label{def_pi}
  An interval procedure $\mathcal P_n$ is a mapping from $ \mathbb D^n$ into the set of all pairs of ordered measurable functions $\{(l,u): l,u:\mathbb R^d \times \{-1,1\} \to \mathbb R, ~ l \leq u\}$. We set
  \begin{equation*}
    (l_{D_n}, ~ u_{D_n}) = \mathcal P_n(D_n).
  \end{equation*}
  We say that $\mathcal P_n$ is a $T$-conditional prediction-interval procedure at level $1-\alpha$, if for all $n \in \mathbb N$ 
  \begin{equation*}
    \mathbb P\left( Y \in [l_{D_n}(X,T), u_{D_n}(X,T)]  ~ | ~ T=t \right ) \geq 1-\alpha, \quad t \in \{-1,1\}.
  \end{equation*}
\end{definition}

The standard conformal method allows to construct unconditional prediction intervals only under the assumption that the data is i.i.d.. In Section \ref{sec_conf_pi}, we show how to modify the conformal method to obtain a $T$-conditional prediction-interval procedure in the sense of Definition \ref{def_pi}. 

\subsection{Theoretical results}

We propose prediction intervals for $\tau(X)$ that either control finite sample coverage or asymptotic coverage.

\begin{definition} \label{def_pi_treat}
    Let $\Gamma_{D_n}: \mathbb R^d \to \{[a,b]: a \leq b\}$ be an interval-valued function that depends on $D_n$. $\Gamma_{D_n}(X)$  is called a prediction interval for the individual treatment effect $\tau(X)$ at level $1-\alpha$, if for all $n \in \mathbb N$
    \begin{equation} \label{ineq1}
        \mathbb P(\tau(X) \in \Gamma_{D_n}(X)) \geq 1-\alpha.
    \end{equation}
    We say that $\Gamma_{D_n}(X)$ is an asymptotically valid prediction interval for $\tau(X)$ at level $1-\alpha$, if
    \begin{equation} \label{ineq2}
        \lim_{n\to\infty} \mathbb P(\tau(X) \in \Gamma_{D_n}(X)) \geq 1-\alpha.
    \end{equation}
\end{definition}

In all of the following results, we assume that we have a prediction $T$-conditional prediction-interval procedure $\mathcal P_n$ in the sense of Definition $\ref{def_pi}$. Construction of such a procedure is described in Section \ref{sec_conf_pi}. We note that the following results hold for any $T$-conditional prediction-interval procedure $\mathcal P_n$ (not necessarily based on the method proposed in Section \ref{sec_conf_pi}). 

The first result only requires the i.i.d. assumption and a $T$-conditional prediction-interval procedure at level $1-\alpha/2$.
\begin{theorem} \label{th_te1}
  Let $\mathcal P_n$ be a $T$-conditional prediction-interval procedure at level $1-\alpha/2$ in the sense of Definition \ref{def_pi}. Set $(l_{D_n},u_{D_n})=\mathcal P_n(D_n)$ and
  \begin{equation} \label{int1}
    \Gamma_{D_n} (X) = \left[l_{D_n}(X,1) - u_{D_n}(X,-1) ,u_{D_n}(X,1) - l_{D_n}(X,-1) \right].
\end{equation}
     Then $\Gamma_{D_n} (X)$ is a prediction interval for $\tau(X)$ at level $1-\alpha$.
\end{theorem}
In this result, $\Gamma_{D_n} (X)$ controls coverage at level $1-\alpha$, but $\mathcal P_n$ is a $T$-conditional prediction-interval procedure at level $1-\alpha/2$. We can improve this result (in terms of interval length) by assuming conditional independence of the errors $e_{X,-1}$ and $e_{X,1}$ given $X$. Under this additional assumption, a $T$-conditional prediction-interval procedure at level $\sqrt{1-\alpha}$ (which is strictly smaller than $1-\alpha/2$ for all $\alpha \in (0,1))$ is sufficient such that inequality $(\ref{ineq1})$ continues to hold.
\begin{theorem} \label{th_te1_b}
  Let $\mathcal P_n$ be a $T$-conditional prediction-interval procedure at level $\sqrt{1-\alpha}$ in the sense of Definition \ref{def_pi}. Assume that $e_{X,-1}$ and $e_{X,1}$ are conditionally independent given $X$. Set $(l_{D_n},u_{D_n})=\mathcal P_n(D_n)$ and let $\Gamma_{D_n} (X)$ be defined as in (\ref{int1}). Then $\Gamma_{D_n} (X)$ is a prediction interval for $\tau(X)$ at level $1-\alpha$.
\end{theorem}
Both, Theorem \ref{th_te1} and \ref{th_te1_b}, have strong finite-sample coverage claims under minimal assumptions. Thus, it is not surprising that these intervals are conservative in certain situations. This is especially the case if the true regression function $f(x,t)$ is simple and easy to estimate (see the simulations of Section \ref{sec_sim}). In Theorem \ref{th_te1} and \ref{th_te1_b} we made no assumption on the accuracy of the estimated regression function $f_{D_n}(x,t)$. This means that these intervals do not only take the variance of the error $e_{X,1} - e_{X,-1}$ but also the estimation error of $f_{D_n}(x,t)$ into account.

The intervals of the next result are much shorter, but the assumptions are considerable stronger and the coverage claim is weaker. 
\begin{theorem} \label{th_te2}
  Let $\mathcal A_n$ a be consistent prediction procedure as given in Definition \ref{def_pred}. Let $\mathcal P_n$ be a $T$-conditional prediction-interval procedure at level $1-\alpha$ as given in Definition \ref{def_pi}. Set $(l_{D_n},u_{D_n})$ and suppose that 
  \begin{equation} \label{cover}
    \mathbb P\left(l_{D_n}(X,T) \leq f_{D_n}(X,T) \leq  u_{D_n}(X,T) \right) ~ = ~ 1.
  \end{equation}
  Assume that conditional on $X$, $(e_{X,1}, e_{X,-1})'$ is multivariate normal-distributed with $\mathrm{Var}(e_{X,1})=\mathrm{Var}(e_{X,-1})$ and non-negative covariance. For $t \in \{-1,1\}$, set
  \begin{align*}
    \tilde l_{D_n}(X,t) &= f_{D_n}(X,t) - \frac{f_{D_n}(X,t)-l_{D_n}(X,t)}{\sqrt{2}}; \\
    \tilde u_{D_n}(X,t) &= f_{D_n}(X,t) + \frac{u_{D_n}(X,t)-f_{D_n}(X,t)}{\sqrt{2}}
  \end{align*}
  and
  \begin{equation*}
    \Gamma_{D_n} (X) = \left[\tilde l_{D_n}(X,1) - \tilde u_{D_n}(X,-1) ,\tilde u_{D_n}(X,1) - \tilde l_{D_n}(X,-1) \right].
  \end{equation*}
    Then $\Gamma_{D_n} (X)$ is an asymptotically valid prediction interval for $\tau(X)$ at level $1-\alpha$.
\end{theorem}
An inspection of the proof and the remark after Lemma \ref{le_a1} shows that if the errors $e_{X,-1}$ and $e_{X,1}$ are negatively correlated then Theorem \ref{th_te2} continuous to hold if we omit the factor $\sqrt{2}$ in the definition of $\tilde l_{D_n}(X,t)$ and $\tilde u_{D_n}(X,t)$. (This means, $\tilde l_{D_n}(X,t)$ and $\tilde u_{D_n}(X,t)$ reduces to $l_{D_n}(X,t)$ and $u_{D_n}(X,t)$, respectively.)

Because $\mathcal A_n$ is a consistent prediction procedure, the proof of the theorem basically reduces to showing an elementary probabilistic inequality for Gaussian random variables which is given in Lemma \ref{le_a1} of the appendix.

\section{$T$-conditional prediction-interval procedures} \label{sec_conf_pi}

Conformal prediction intervals control coverage marginally \citep{vovk2005}. But it is straightforward to construct conformal prediction intervals with coverage conditional on any event with positive probability. Loosely speaking, one only needs to apply the conformal method "conditional" on the event of interest. For completeness, this section describes in detail how to apply the conformal method to construct $T$-conditional prediction-interval procedures $\mathcal P_n$ in the sense of Definition \ref{def_pi}. Because computation of conformal prediction intervals is expensive and infeasible in certain situations, we also show how to adapt the ``split-conformal'' procedure (\cite{lei2018}). Computation of the ``split-conformal'' prediction intervals is much cheaper but the resulting intervals are typically longer. 

\subsection{Nonconformity measures and scores}

In the following, we define the terms conformal procedure, nonconformity measure and nonconformity score.
\begin{definition} \label{def_nonconformity}
  A conformal procedure $\mathcal M_n$ is a mapping from $\mathbb D^n$ into the set of measurable functions $\{m: \mathbb D \to [0,\infty)\}$. We set
  \begin{equation*}
    m_{D_n} = \mathcal M_n(D_n),
  \end{equation*}
  where we call $m_{D_n}$ the nonconformity measure and $m_{D_n}(x,y,t)$ the nonconformity score of the observation $(x,y,t) \in \mathbb D$.

  $\mathcal M_n$ is called permutation invariant if, for any $n \in \mathbb{N}$, for any permutation $\sigma_n: \{1,\dots,n\} \to \{1,\dots,n\}$ and for any $(x,y,t) \in \mathbb D$,
  \begin{equation*}
    \mathbb P \left( m_{D_n}(x,y,t) = m_{D_{n,\sigma_n}}(x,t) \right) ~ = ~ 1.
  \end{equation*}
  where $D_{n,\sigma_n}=((X_{\sigma_n(1)},Y_{\sigma_n(1)},T_{\sigma_n(1)}),\dots,(X_{\sigma_n(n)},Y_{\sigma_n(n)},T_{\sigma_n(n)}))$.
\end{definition}
In the following, we give examples of different nonconformity measures.
\begin{example}[Absolute residuals]
  The simplest nonconformity measure is given by
  \begin{equation*}
    m_{D_n}(x,y,t) = | y - f_{D_n}(x,t)|.
  \end{equation*}
  However, this measure is not optimal if the conditional variance of $Y$ given $X$ and $T$ is not constant, i.e., if the error $e_{X,T}$ is heteroscedastic.
\end{example}
\begin{example}[Absolute standardized residuals]
  To take heteroscedasticity into account, one can define the nonconformity measure as
  \begin{equation*}
    m_{D_n}(x,y,t) = | y - f_{D_n}(x,t)| / \sigma_{D_n}(x,t),
  \end{equation*}  
  where $\sigma_{D_n}^2(x,t)$ is an estimator for the conditional variance of $Y$ conditional on $X$ and $T$. Typically, a kernel-based algorithm is used to obtain $\sigma_{D_n}^2(x,t)$.
\end{example}
\begin{example}[ML-based measures]
  Modern machine learning algorithms, such as neural networks, are able to estimate the distribution of the error $e_{X,t}$ and one can define nonconformity measures tailored to these algorithms as shown in \cite{kivaranovic2019b,romano2019}.
\end{example}
Note that the conformal procedure in these examples are permutation invariant, if the corresponding estimation procedures are permutation invariant. 

\subsection{$T$-conditional conformal prediction intervals} \label{sec_full}

%We show how to adapt the conformal method to obtain a $T$-conditional prediction-interval procedure in the sense of Definition $\ref{def_pi}$. 
Let $D_{n+1}(y)$ denote the augmented dataset which contains $D_n$ and the pseudo observation $(X,y,T)$, where $y \in \mathbb R$. Recall that
\begin{equation*}
  m_{D_{n+1}(y)} = \mathcal M_{n+1}(D_{n+1}(y)),
\end{equation*}
is the nonconformity measure of the conformal procedure $\mathcal M_{n+1}$ evaluated at $D_{n+1}(y)$. Denote by $D_{n+1,T}(y)$ the subset of observations in $D_{n+1}(y)$ with treatment equal to $T$. Note that $D_{n+1,T}(y)$ is non-empty because $(X,y,T)$ is always contained. For ease of notation, we set
\begin{equation} \label{eq_vector}
  \Scale[0.85]{
  m_{D_{n+1}(y)}(D_{n+1,T}(y)) = \left \{ m_{D_{n+1}(y)}((\tilde X, \tilde Y, \tilde T)): (\tilde X, \tilde Y, \tilde T) \in D_{n+1,T}(y), ~ \tilde T = T \right \},
  }
\end{equation}
this means, $m_{D_{n+1}(y)}(D_{n+1,T}(y))$ is the set of nonconformity scores of the nonconformity measure $m_{D_{n+1}(y)}$ evaluated on the random variables in $D_{n+1,T}(y)$. Denote by $R_{D_{n+1}(y)}(y)$ the rank of $m_{D_{n+1}(y)}(X,y,T)$ among the random variables in $m_{D_{n+1}(y)}(D_{n+1,T}(y))$. In case of ties, we use a data-independent tie-breaking rule. Let $N_T = \sum_{i=1}^n \mathds{1}_{\{T_i=T\}}$ and set 
\begin{equation} \label{eq_pi_full}
  \Gamma^{F}_{D_n}(X,T) = \left\{y \in \mathbb R: ~ R_{D_{n+1}(y)}(y) \leq  \left \lceil (1-\alpha) \left (N_T +1 \right) \right \rceil \right\}
\end{equation}
(Note that if $\alpha (N_T +1)<1$, then $\Gamma^{F}_{D_n}(X,T) = (-\infty,\infty)$.) The superscript $F$ emphasizes that the full dataset is used for computation of the nonconformity measure. This is in contrast to the ``split-conformal'' version of the next section, where only a subset is used for computation of the nonconformity measure.
\begin{theorem} \label{th_full}
  Let $\mathcal M_n$ be a permutation invariant conformal procedure as given in Definition \ref{def_nonconformity} . Let $\Gamma^F_{D_n}(X,T)$ be defined as in (\ref{eq_pi_full}). Let $t \in \{-1,1\}$ and set $p_t = \mathbb P(T=t)$. Then  \begin{equation*}
    1-\alpha ~ \leq ~ \mathbb{P}\left(Y \in \Gamma^{F}_{D_n}(X,T) ~ | ~ T=t \right) ~ \leq ~ 1-\alpha + \frac{1-(1-p_t)^{n+1}}{(n+1)p_t}.
  \end{equation*}
\end{theorem}
The upper bound in this theorem converges to $1-\alpha$ as $n \to \infty$, therefore, the procedure is asymptotically exact. We note that $\Gamma^F_{D_n}(X,T)$ is not an interval in general. But, obviously, the interval
\begin{equation*}
  \left[\inf ~ \Gamma^F_{D_n}(X,T) , \sup ~ \Gamma^F_{D_n}(X,T) \right]
\end{equation*}
is a $T$-conditional prediction interval for $Y$, because it is a superset of $\Gamma^F_{D_n}(X,T)$. Computation of $\Gamma^F_{D_n}(X,T)$ is very expensive if a brute-force approach is used because then one needs to compute the nonconformity measure $m_{D_{n+1}(y)}$ for all augmented datasets $D_{n+1}(y)$, where $y \in \mathbb R$. Of course, practically, we do this only on a sufficiently dense grid for $y$, but if the dataset is large and one is using a modern learning algorithm, such as a neural network, this is computationally infeasible.

\subsection{$T$-conditional ``split-conformal'' prediction intervals} \label{sec_split}

For $1 \leq m < n$, $D_m$ is the dataset of the first $m$ observation and denote by $D_{m:n}$ the dataset of the remaining $n-m$ observations. Let $D_{m:(n+1)}(y)$ be the augmented dataset which contains $D_{m:n}$ and the pseudo observation $(X,y,T)$. Denote by $D_{m:(n+1),T}(y)$ the subset of observations in $D_{m:(n+1)}(y)$ with treatment equal to $T$. Again, $D_{m:(n+1),T}(y)$ is non-empty because $(X,y,T)$ is always contained. Let $m_{D_m}(D_{m:(n+1),T}(y))$ be defined as in (\ref{eq_vector}) with $D_m$ replacing $D_{n+1}(y)$ and $D_{m:(n+1),T}(y)$ replacing $D_{n+1,T}(y)$. This means, $m_{D_m}(D_{m:(n+1),T}(y))$ is the set nonconformity scores of the nonconformity measure $m_{D_m}$ evaluated on the random variables in $D_{m:(n+1),T}(y)$. Denote by $R_{D_m}(y)$ the rank of $m_{D_m}(X,y,T)$ among the random variables in $D_{m:(n+1),T}(y)$. In case of ties, we use a data-independent tie-breaking rule. Let $\tilde N_T = \sum_{i=m+1}^n  \mathds{1}_{\{T_i=T\}}$ and set
\begin{equation} \label{eq_pi_split}
  \Gamma^{S}_{D_n}(X,T) = \left\{y \in \mathbb R: ~ R_{D_m}(y) \leq  \left \lceil (1-\alpha) \left (\tilde N_T +1 \right) \right \rceil \right\}  
\end{equation}
(Note that if $\alpha(\tilde N_T +1) < 1$ then $\Gamma^{S}_{D_n}(X,T)=(-\infty,\infty)$.) The superscript $S$ emphasizes that only a subset of the data is used for computation of the nonconformity measure.
\begin{theorem} \label{th_split}
  Let $\mathcal M_n$ be a conformal procedure as given in Definition \ref{def_nonconformity} . Let $\Gamma^S_{D_n}(X,T)$ be defined as in (\ref{eq_pi_split}). Let $t \in \{-1,1\}$ and set $p_t = \mathbb P(T=t)$,
  \begin{equation*}
    1-\alpha ~ \leq ~ \mathbb{P}\left(Y \in \Gamma^{S}_{D_n}(X,T) ~ | ~ T=t \right) ~ \leq ~ 1-\alpha + \frac{1-(1-p_t)^{n-m+1}}{(n-m+1)p_t}.
  \end{equation*}
\end{theorem}
Again, the upper bound converges to $1-\alpha$ as $n-m \to \infty$, therefore, the procedure is asymptotically exact. Note that in the ``split-conformal'' case, we do not need that $\mathcal M_n$ is permutation invariant. Computation of $\Gamma^S_{D_n}(X,T)$ is much cheaper, because we need to calculate the nonconformity measure only once on $D_m$. The set $\Gamma^{S}_{D_n}(X,T)$ is an interval if the nonconformity measure $m_{D_n}$ is convex in $y$, which is typically the case for the most common nonconformity measures.

\section{Simulation} \label{sec_sim}

We fixed the covariate dimension $d=10$ and the significance level $\alpha=0.1$. Further, we set $\mathbb P(T=1)=1/2$. We considered 4 different sample sizes: $n=300,700,1200,2000$. The covariate vector $X$ is $N(0,\Sigma)$-distributed, where $\Sigma$ is a $10 \times 10$ matrix where each element on the diagonal is equal $1$ and each element off the diagonal is equal to $\rho$. We considered two values for $\rho$: $0.2$ (low correlation) and $0.8$ (high correlation). We considered two different regression functions:
\begin{equation*}
  f_1(x,t) = x_1 + x_2 + x_3 + t \quad \text{and} \quad f_2(x,t) = \mathrm{sign}(f_1(x,t))  f_1(x,t)^2.
\end{equation*}
We also considered two different distributions for the error distribution: $e_{X,t}$ conditional on $X$ is either a standard normal distributed or Laplace distributed with variance $1$. Considering all possible combinations, we have 8 different data generating processes in total: LOW COR vs. HIGH COR, LINEAR vs. NON LINEAR and NORMAL vs. LAPLACE.

We also considered two different $T$-conditional prediction interval procedures: The first is the linear least squares method with the absolute error as nonconformity measure (as described in Example 1). Because linear models are very fast to fit, we use the full conformal version as described in Section \ref{sec_full} to construct the prediction intervals. The second is a neural network with one hidden layer and ten hidden nodes and the same nonconformity measure as in the previous case. Because fitting a neural network is computationally more expensive, we use the split conformal approach as described in Section \ref{sec_split}. We used 2/3 of the data for training of the model and 1/3 for computation of the nonconformity scores. For both methods, we constructed the prediction interval as described in Theorem \ref{th_te1_b} and \ref{th_te2}. This means, in total we have 4 different prediction intervals which we label LM1, LM2, NN1 and NN2. Here LM denotes the linear least squares method, NN the neural network and 1 and 2 denote the construction as described in Theorem \ref{th_te1_b} and \ref{th_te2}, respectively.

For each of the $4$ sample sizes, $8$ data generating processes and $4$ prediction intervals, we repeated the following steps $1000$ times: First, we draw the dataset $D_{n}$ and the new covariate vector $X$ and the prediction interval for $\tau(X)$. Then we draw $\tau(X)$. And, finally, we compute the length of the prediction interval and check if it covers $\tau(X)$. The length of the prediction intervals is computed relative to the length of the oracle which knows the underlying data generating process. To compute the conformal prediction intervals, we took a sufficiently large grid with step size $.1$ and defined the interval as the smallest and largest value in the grid that satisfied the condition (\ref{eq_pi_full}). The results are summarized in Figure 1 and 2. 

In Figure 1, all data generating processes with Gaussian error distribution are shown. We see that the methods LM1 and NN1 always control coverage. This is in line with Theorem \ref{th_te1_b} which guarantees finite sample coverage in all these situations. But, because of the strong coverage claim, these intervals are also very conservative, especially when the true regression function is a simple linear model (first and second row). Furthermore, we see that the method NN1 significantly outperforms LM1 (in terms of length) when the true model is non linear (third and fourth row) and the sample size grows. This means, the neural network is able to estimate the non linear regression function more accurately as the sample size grows and, therefore, gives more accurate intervals. The linear regression is clearly not able to do this.

On the other hand, the methods LM2 and NN2 are much closer to the nominal level of 90\% and their lengths are much closer to the oracle. But we note that the assumptions of Theorem \ref{th_te2} are much stronger. We see that LM2 is nearly optimal (in terms of coverage and length) if the true regression function is linear. The method NN2 performs only marginally worse in these cases. This is not surprising because both, linear regression and neural networks, are able to consistently estimate linear functions. NN2 performs a bit worse because the complexity of the neural network is not required for the simple structure of the linear regression function. In the cases where the true regression function is non linear (third and fourth row), the method LM2 has coverage consistently below the nominal level. Furthermore, the length of the intervals is up to 4 times larger than of the oracle and does not improve with growing sample size. Clearly, LM2 does not satisfy the assumptions of Theorem \ref{th_te2} and there for it is not surprising that it does not control coverage. However, NN2 performs much better in these situations. The coverage of NN2 is close to the nominal level and the length converges to the length of the oracle as the sample size grows. This shows that the gain of complex models is very high if sufficient data is available and the true regression function not linear. NN2 controls not only the coverage more accurately, the intervals are also shorter by a factor 3 to 4.

To demonstrate the robustness of our results, Figure 2 shows the same data generating processes with Laplace error instead of Gaussian error. All four methods behave very similarly to the Gaussian case.  This suggests that Theorem \ref{th_te2} may be valid for a wider class of distributions.

\section{Discussion} \label{sec_discussion}

In this manuscript we proposed several methods to construct prediction intervals for the estimation of individual treatment effects based on complex prediction models such as regression models with variable selection or neural network algorithms. The prediction intervals give bounds for the differences between the outcomes under treatment and control for individual patients. This is in contrast to confidence intervals for the predicted individual treatment effect, which estimate the expected  individual treatment effect for patients with specific baseline characteristics. 

The intervals of Theorem \ref{th_te1} and \ref{th_te1_b} control coverage of the individual treatment effect under minimal assumptions. Theorem \ref{th_te2} provides considerably shorter intervals that are valid under stronger assumptions. A limitation in the interpretation of conformal prediction interval is that the coverage probabilities are guaranteed on a population level only, but not conditional on a specific vector of covariates. Especially, they are only meaningful, if predictions are made for a population with the same covariate distribution as the population from which the intervals were computed. This is in contrast to the classical prediction intervals, e.g., for linear models, that also provide conditional coverage, given the value of the covariates. 

%The width of the prediction intervals is determined by the uncertainties of the prediction model, as, e.g., the standard errors of the model coefficients in a regression model, and the variability in the outcomes that remains after adjustment for all baseline variables. Therefore, an essential difference between the prediction intervals and the confidence intervals that the width of the prediction intervals does not converge to 0 as the sample size increases: while the variability of the model estimates decreases with the sample size, the  variability that cannot be explained by the covariates remains.

The width of the prediction intervals can be a useful measure to compare different prediction models. While more flexible models may be able to explain more variability by adjusting for more  covariates or more complex function of the covariates, they also introduce variability because of increasing uncertainties in the model estimates for more extensive models.  As demonstrated in the simulation study, if the sample size is low  simpler models, even though they are  miss-specified, can lead to narrower prediction intervals than more complex models which overfit the data and lead to highly variable model estimates.

\renewcommand{\floatpagefraction}{0.1}

\begin{figure*}
\centering
\begin{minipage}{.5\textwidth}
  \centering
  \centerline{\includegraphics[trim=0 40 0 0, clip, width=1\textwidth]{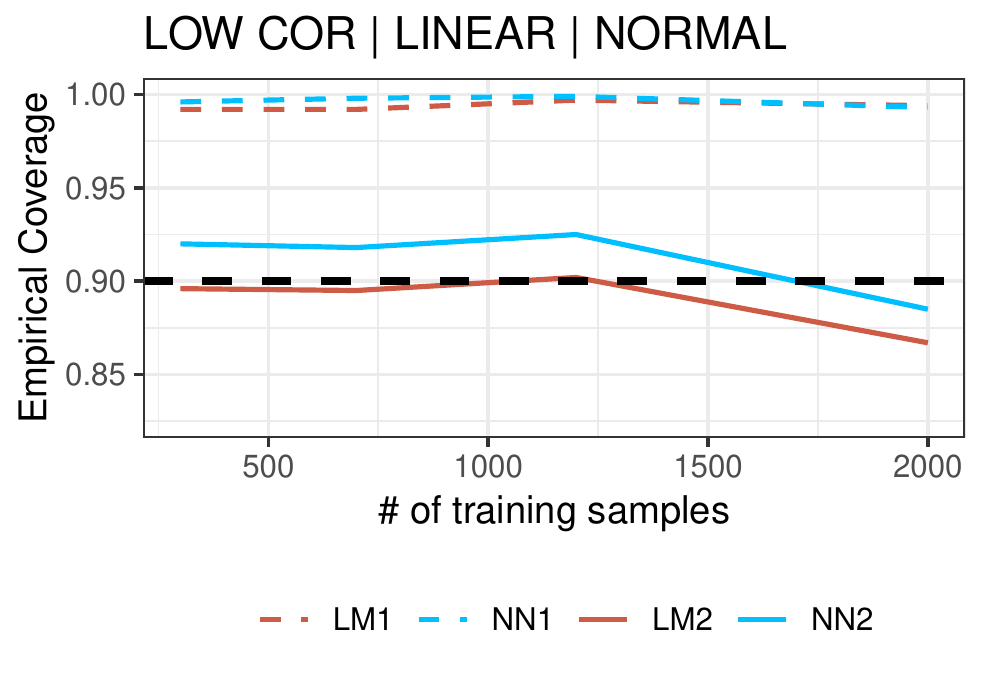}} 
  \centerline{\includegraphics[trim=0 40 0 0, clip, width=1\textwidth]{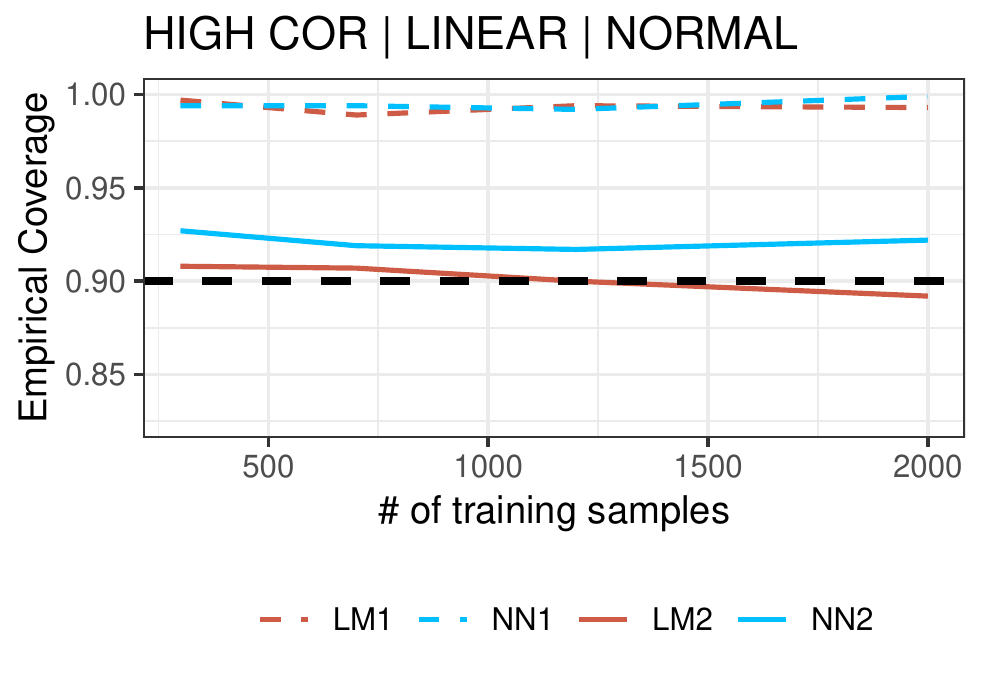}}
  \centerline{\includegraphics[trim=0 40 0 0, clip, width=1\textwidth]{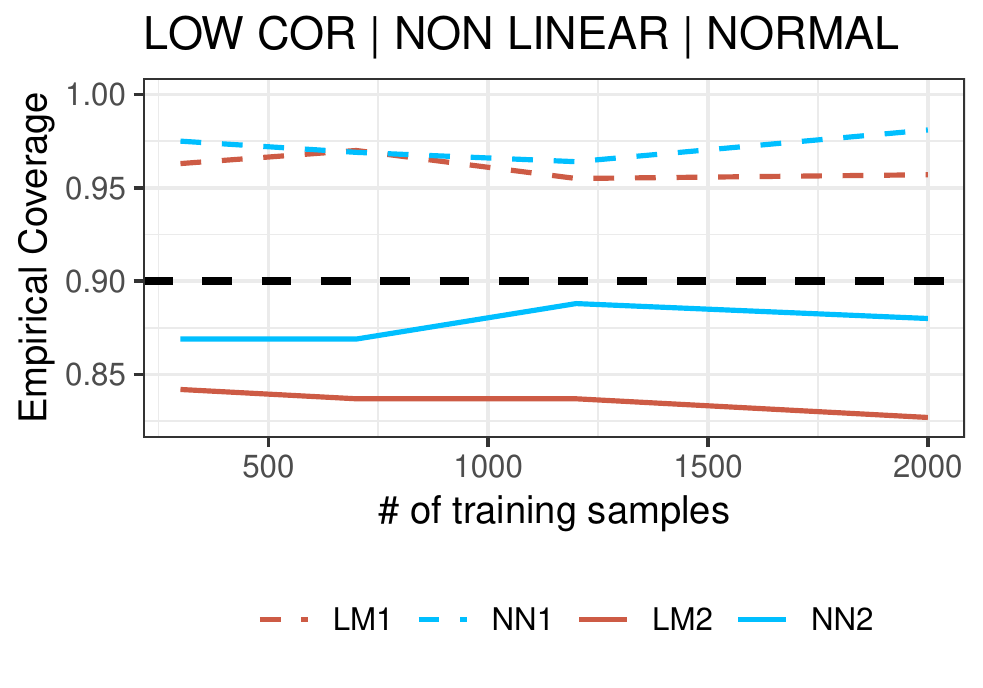}}
  \centerline{\includegraphics[trim=0  0 0 0, clip, width=1\textwidth]{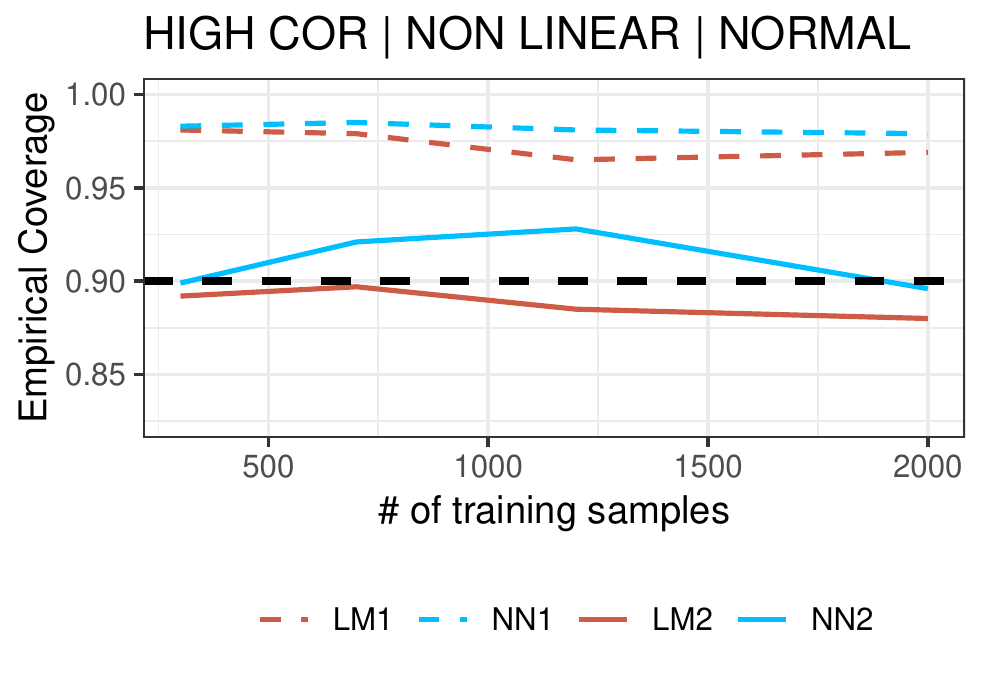}}
\end{minipage}%
\begin{minipage}{.5\textwidth}
  \centering
  \centerline{\includegraphics[trim=0 40 0 0, clip, width=1\textwidth]{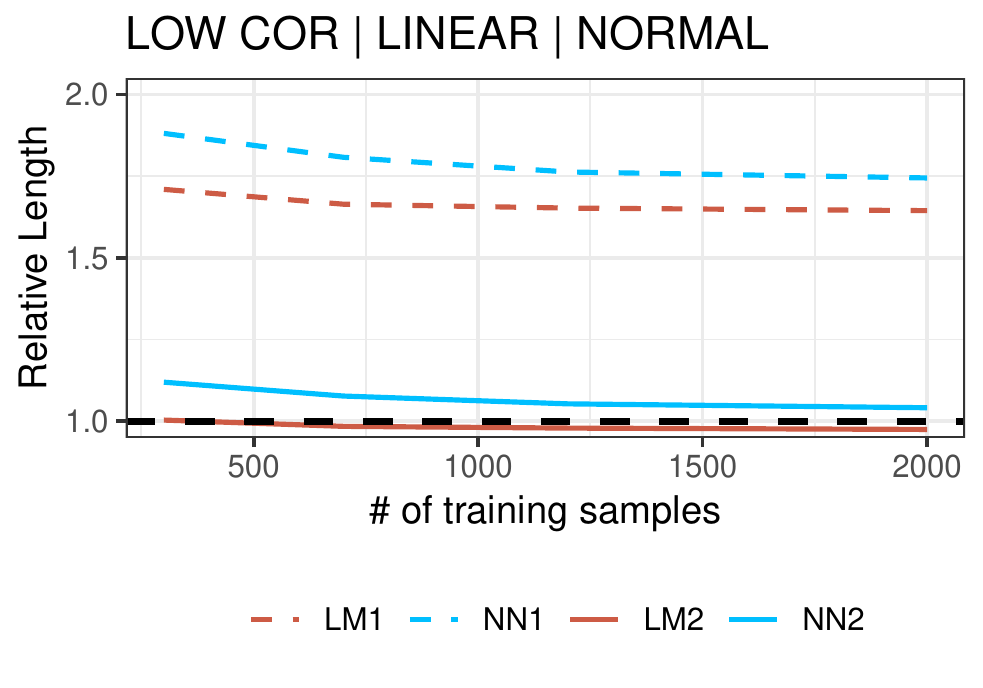}}
  \centerline{\includegraphics[trim=0 40 0 0, clip, width=1\textwidth]{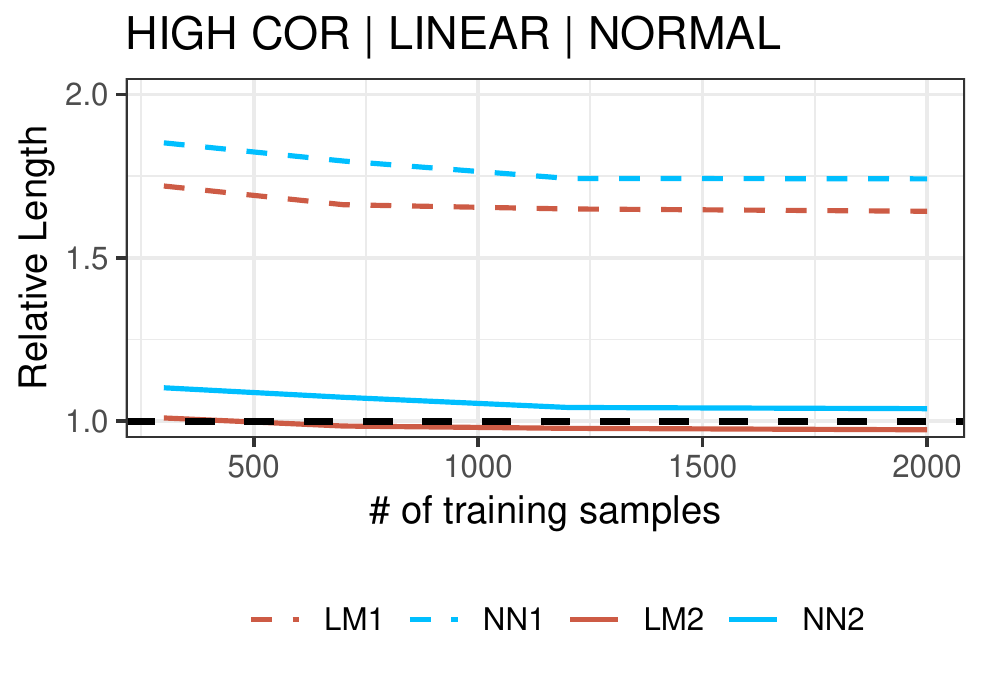}}
  \centerline{\includegraphics[trim=0 40 0 0, clip, width=1\textwidth]{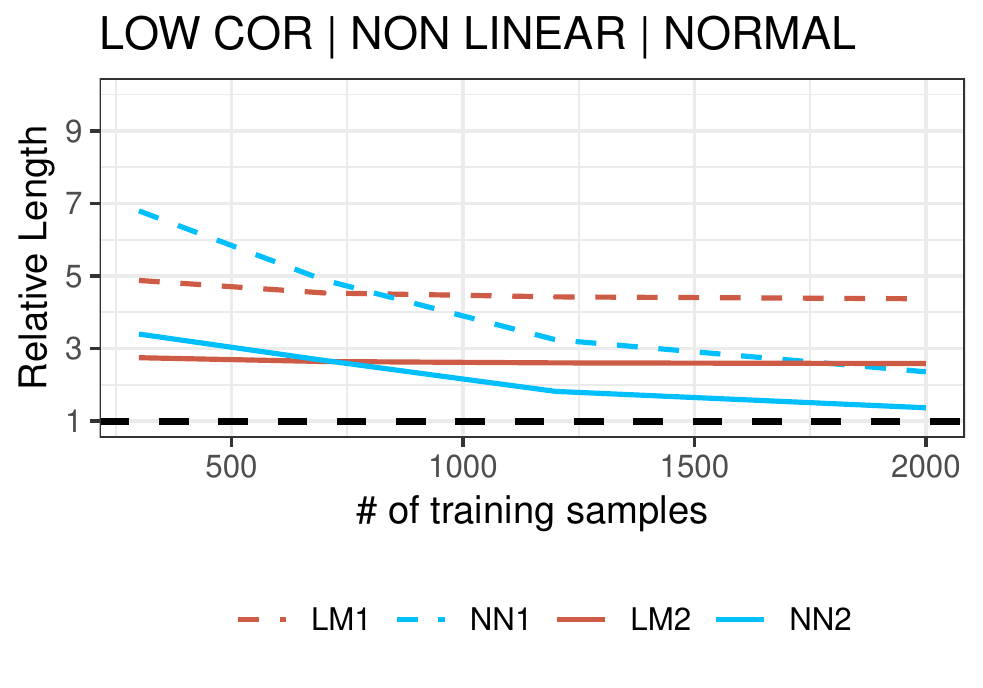}}
  \centerline{\includegraphics[trim=0  0 0 0, clip, width=1\textwidth]{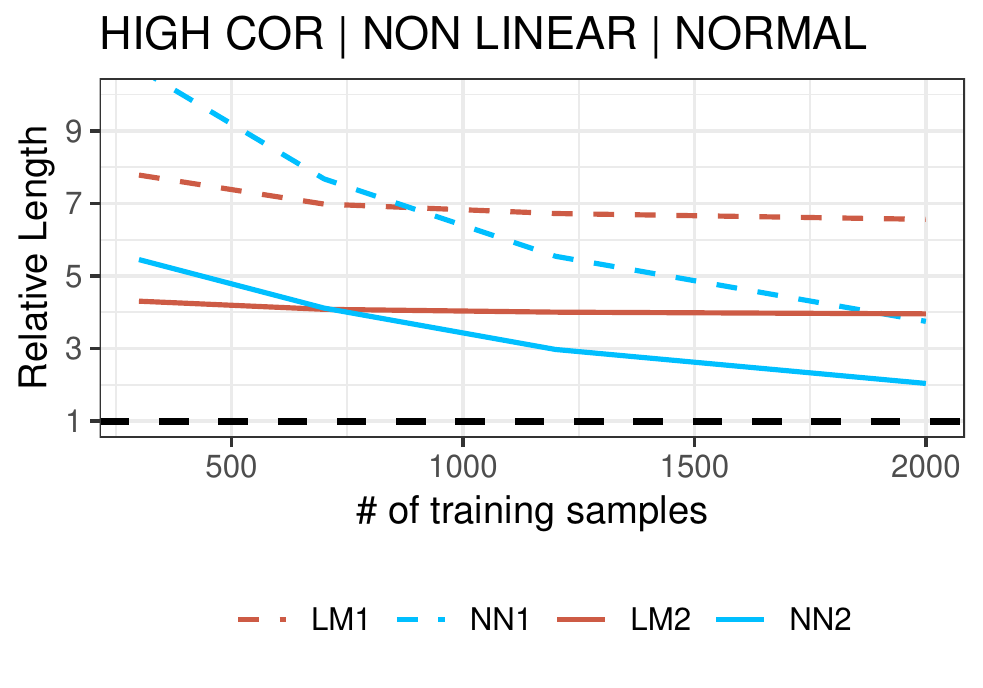}}
\end{minipage}
\caption{Empirical coverage is shown on the l.h.s. (standard error is approximately 1\%). The dotted line at .9 denotes the nominal confidence level. On the r.h.s., the relative length compared to the oracle is shown. The dotted line at 1 denotes the length of the oracle. The data generating process is given above each graph. Each setting is simulated with Gaussian errors.}
\end{figure*}
\renewcommand{\floatpagefraction}{0.1}

\begin{figure*}
\centering
\begin{minipage}{.5\textwidth}
  \centering
  \centerline{\includegraphics[trim=0 40 0 0, clip, width=1\textwidth]{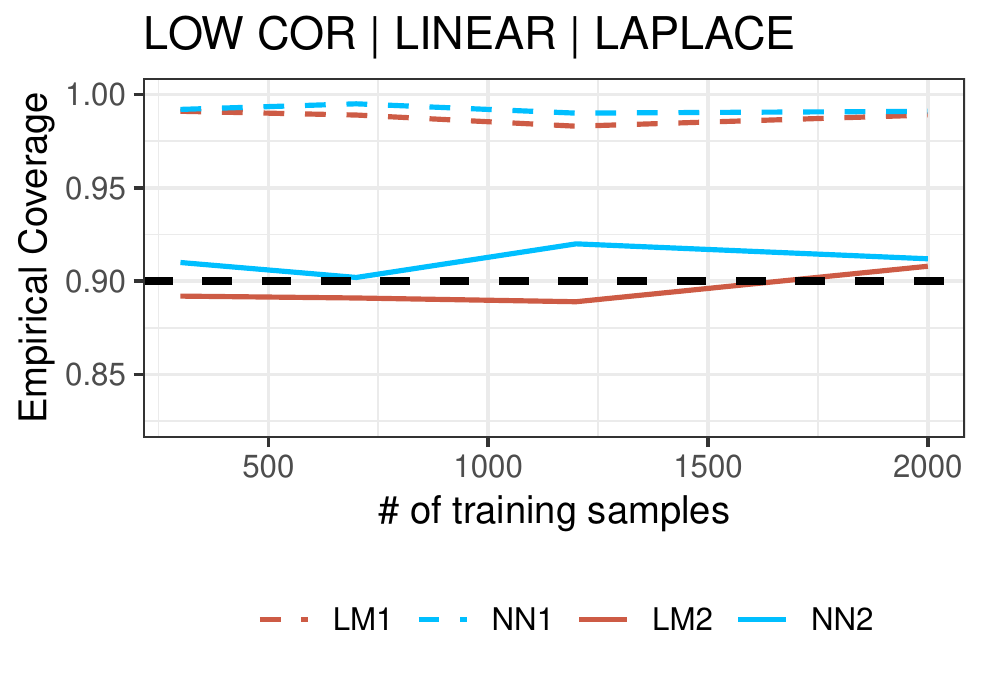}}
  \centerline{\includegraphics[trim=0 40 0 0, clip, width=1\textwidth]{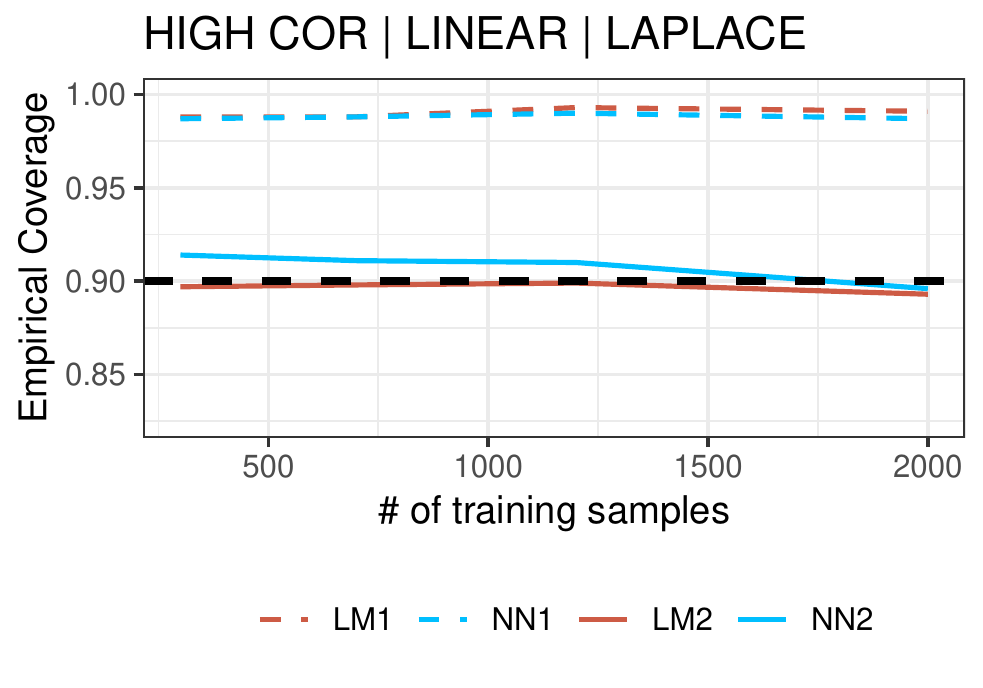}}
  \centerline{\includegraphics[trim=0 40 0 0, clip, width=1\textwidth]{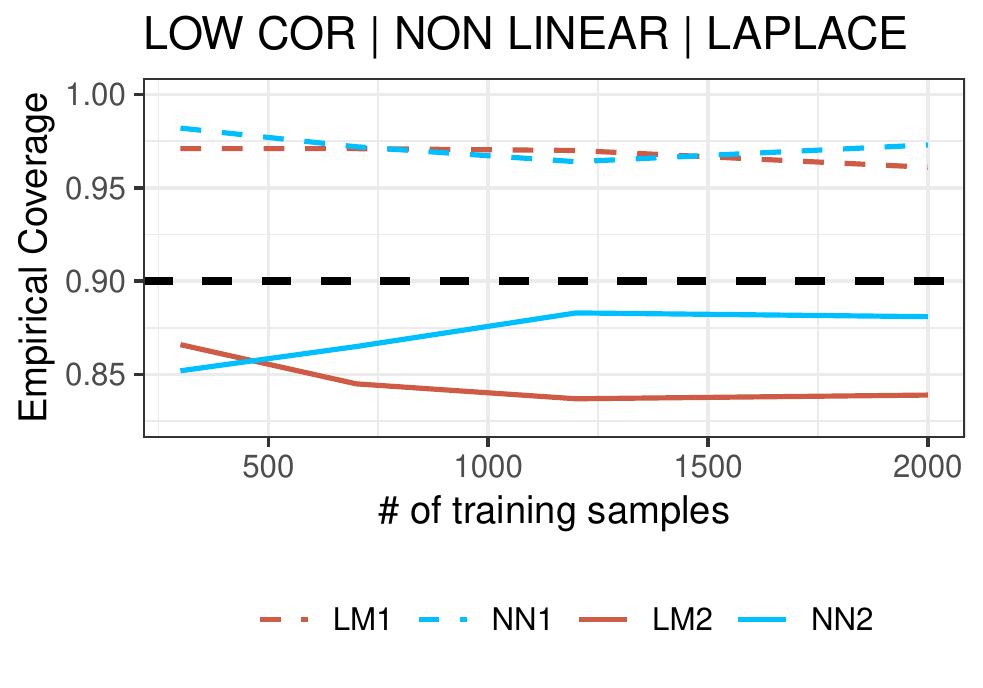}}
  \centerline{\includegraphics[trim=0  0 0 0, clip, width=1\textwidth]{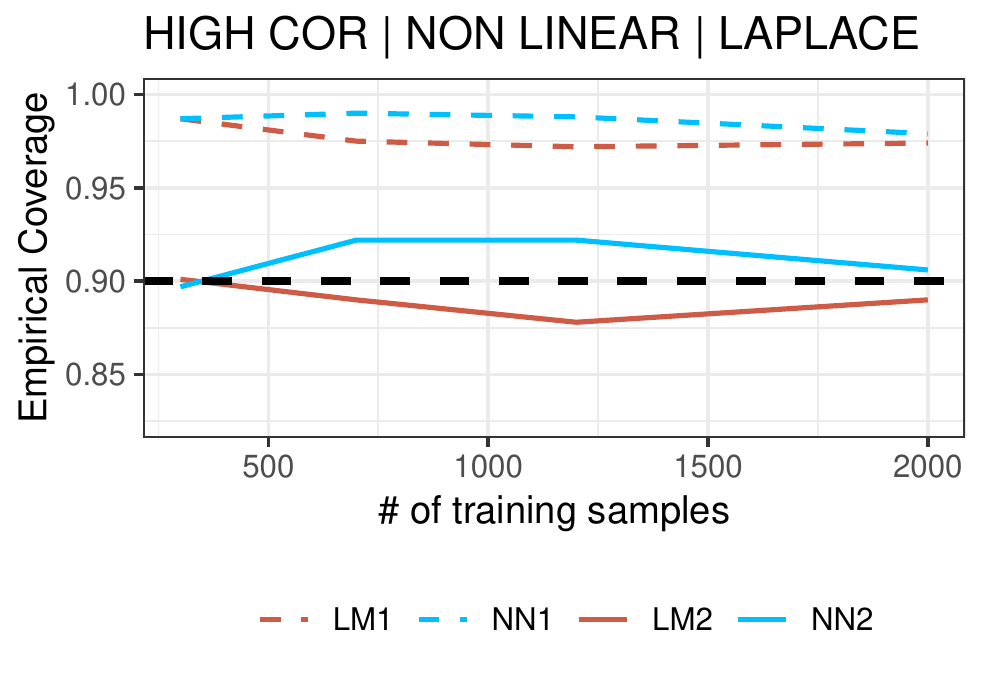}}
\end{minipage}%
\begin{minipage}{.5\textwidth}
  \centering
  \centerline{\includegraphics[trim=0 40 0 0, clip, width=1\textwidth]{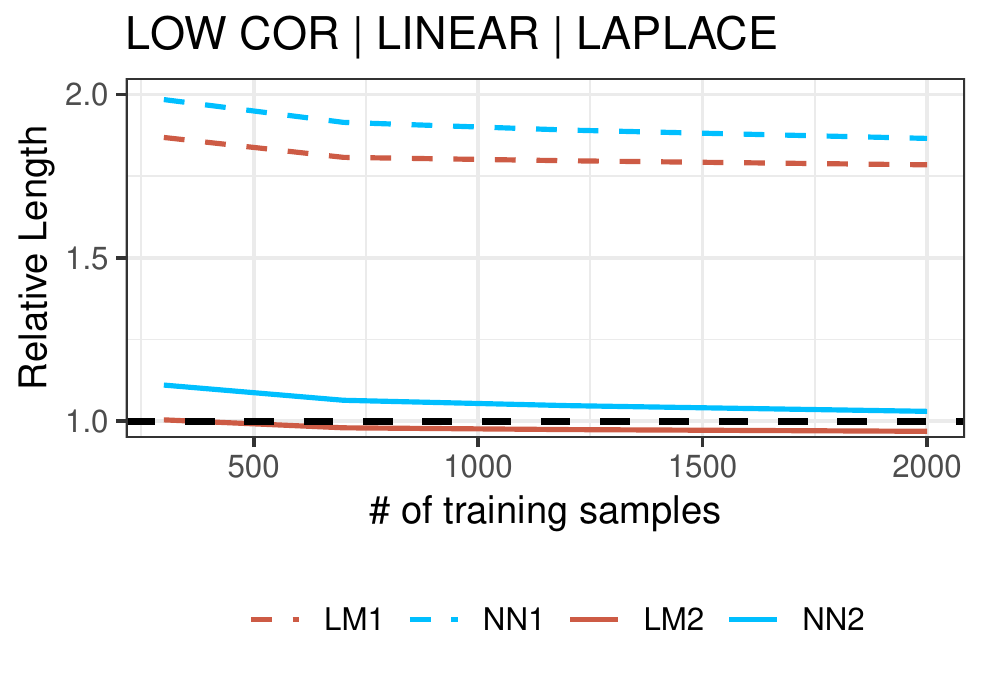}}
  \centerline{\includegraphics[trim=0 40 0 0, clip, width=1\textwidth]{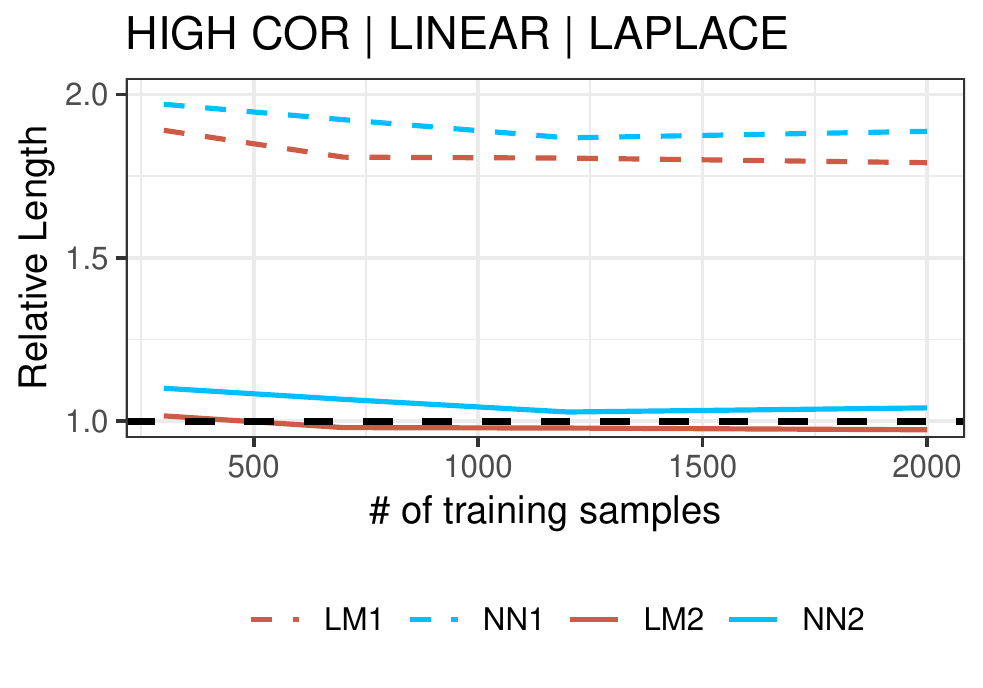}}
  \centerline{\includegraphics[trim=0 40 0 0, clip, width=1\textwidth]{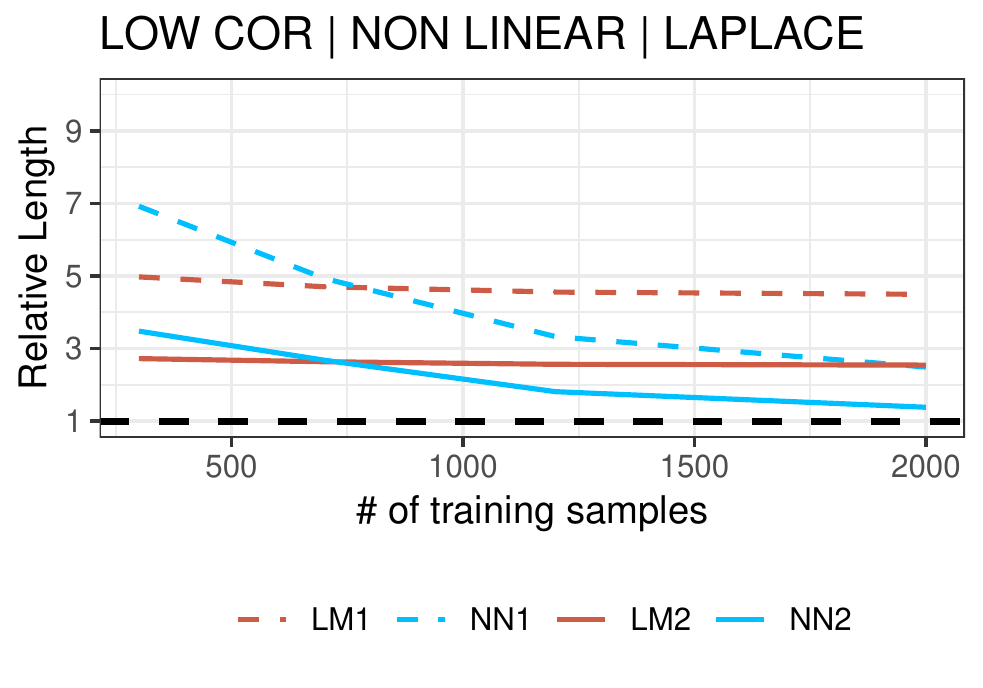}}
  \centerline{\includegraphics[trim=0  0 0 0, clip, width=1\textwidth]{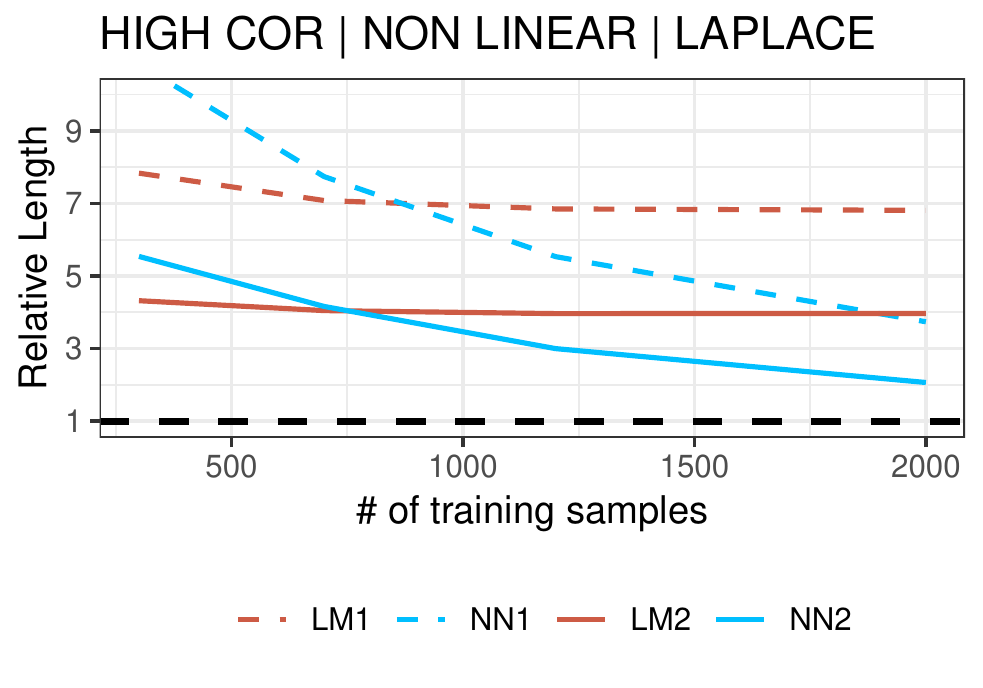}}
\end{minipage}
\caption{The results of the same simulation settings as in Figure 1 are shown but with Laplace error distribution instead of Gaussian error distribution.}
\end{figure*}

\newpage

\begin{appendices}

\section{Proofs of the results} \label{sec_proof}

\subsection{Proof of Theorem \ref{th_te1} and \ref{th_te1_b}}

\begin{proof}[Proof of Theorem \ref{th_te1}]

  By definition of $\Gamma_{D_n}(X)$, we have
  \begin{align*}
    & \left\{ f(X,-1) + \epsilon_{X,-1} \in [l_{D_n}(X,-1), u_{D_n}(X,-1)] \right\} \\ 
    &\cap ~ \left\{ f(X,1) + \epsilon_{X,1} \in [l_{D_n}(X,1), u_{D_n}(X,1)] \right\} \\
    &\subseteq \left\{ \tau(X) \in \Gamma_{D_n}(X) \right\}.
  \end{align*}
  This means, it is sufficient to show that the probability of the smaller event is larger than $1-\alpha$. For $t \in \{-1,1\}$, we have
  \begin{align*}
    &\mathbb P \left(Y \in [l_{D_n}(X,T), u_{D_n}(X,T)] ~ | ~ T=t \right) \\
    &= \mathbb P\left( f(X,t) + \epsilon_{X,t} \in [l_{D_n}(X,t), u_{D_n}(X,t)] ~ | ~ T=t \right) \\
    &= \mathbb P\left( f(X,t) + \epsilon_{X,t} \in [l_{D_n}(X,t), u_{D_n}(X,t)] \right),
  \end{align*}
  where the last equation follows because $T$ is independent of $D_n$ and $X$. If $\mathcal P_n$ is a $T$-conditional prediction interval procedure at level $1-\alpha/2$, the previous equation implies that 
  \begin{equation*}
    \mathbb P\left( f(X,t) + \epsilon_{X,t} \in [l_{D_n}(X,t), u_{D_n}(X,t)] \right) > 1-\alpha/2
  \end{equation*}
  for $t \in \{-1,1\}$. Hence the probability of the intersection is bounded from below by $1-\alpha$.

\end{proof}

\begin{proof}[Proof of Theorem \ref{th_te1_b}]
  Set
  \begin{equation*}
    A_t = \{f(X,t) + \epsilon_{X,t} \in [l_{D_n}(X,t), u_{D_n}(X,t)]\}.
  \end{equation*}
  Because of the observation at the beginning of the previous proof, it is sufficient to show that
  \begin{equation*}
    \mathbb P(A_{-1} \cap A_{1}) = \mathbb E \left[\mathbb P\left(A_{-1} \cap A_{1} ~ | ~ X, D_n \right) \right].
  \end{equation*}
  is bounded from below by $1-\alpha$. Observe that, conditional on $X$ and $D_n$, the only random quantities in $A_{-1}$ and $A_1$ are $e_{X,-1}$ and $e_{X,1}$, respectively. Because $e_{X,-1}$ and $e_{X,1}$ are conditionally independent given $X$, it follows that
  \begin{equation*}
    \mathbb P \left ( A_{-1} \cap A_{1} ~ | ~ X, D_n  \right ) = \mathbb P \left (A_{-1} ~ | ~ X, D_n  \right) \mathbb P \left (A_{1} ~ | ~ X, D_n  \right )
  \end{equation*}
  and therefore it is sufficient to show that
  \begin{equation*}
    \mathbb E \left[\mathbb P\left(A_{-1}~ | ~ X, D_n \right) \mathbb P\left( A_{1} ~ | ~ X, D_n \right) \right].
  \end{equation*}
  is bounded from below by $1-\alpha$. Because the product function is convex on $[0,1]^2$, the Jensen inequality implies that the previous expression is bounded from below by
  \begin{equation*}
        \mathbb E \left[\mathbb P\left(A_{-1}~ | ~ X, D_n \right) \right] \mathbb E \left [\mathbb P\left( A_{1} ~ | ~ X, D_n \right) \right] = \mathbb P(A_{-1}) \mathbb P(A_1) \geq 1-\alpha. 
  \end{equation*}
\end{proof}

\subsection{Proof of Theorem \ref{th_te2}}

Before we prove the theorem, we need the following two lemmas.

\begin{lemma} \label{le_a1}
  Let $(W_1,W_2)'$ be a 2-dimensional centered Gaussian random vector with $\mathrm{Var}(W_1)=\mathrm{Var}(W_2)$ and non-negative covariance. For any $l_1,l_2 \leq 0$ and $u_1,u_2 \geq 0$, we have
  \begin{align*}
    \mathbb P \left(\frac{l_1 - u_2}{\sqrt{2}}  \leq W_1 - W_2 ~ \leq ~ \frac{u_1 - l_2}{\sqrt{2}}\right) ~ &\geq ~ \frac{1}{2} ~ \mathbb P \left(l_1 \leq W_1 \leq u_1 \right) \\
    &+ ~ \frac{1}{2} ~ \mathbb P \left(l_2 \leq W_2 \leq u_2 \right).
  \end{align*}
\end{lemma}
\begin{proof}
  We set $\sigma^2 = \mathrm{Var}(W_1)=\mathrm{Var}(W_2)$. Because the covariance of $W_1$ and $W_2$ is non-negative, it follows that the variance of $W_1-W_2$ is bounded from above by $2 \sigma^2$. This implies that the probability on the l.h.s. of the inequality of the lemma
  is bounded from below by
  \begin{equation*}
    \Phi\left(\frac{u_1-l_2}{2 \sigma} \right) - \Phi\left(\frac{l_1-u_2}{2 \sigma} \right) = \Phi\left(\frac{u_1-l_2}{2 \sigma} \right) + \Phi\left(\frac{u_2-l_1}{2 \sigma} \right) - 1,
  \end{equation*}
  where $\Phi(x)$ is the cumulative distribution function of the standard normal distribution. Because $\Phi(x)$ is concave on the non-negative real numbers and $u_1$ and $-l_1$ are non-negative, it follows that
  \begin{equation*}
    \Phi\left(\frac{u_1-l_2}{2 \sigma} \right) \geq \frac{1}{2}\Phi\left( \frac{u_1}{\sigma} \right) + \frac{1}{2}\Phi\left( \frac{-l_2}{\sigma} \right).
  \end{equation*}
  and likewise for the second summand on the r.h.s. of the previous equation. But this implies the inequality of the lemma.
\end{proof}
If the random variables $W_1$ and $W_2$ are negatively correlated, an inspection of the proof shows that the lemma continuous to hold if we replace $(l_1-u_2) / \sqrt{2}$ and $(u_1-l_2)/ \sqrt{2}$ by $l_1-u_2$ and $u_1-l_2$ on the l.h.s. of the inequality, respectively.

\begin{lemma} \label{le_a2}
  Let $W$ be a centered Gaussian distribution. Let $l,u \in \mathbb{R}$ with $l \leq u$. Then for all $\epsilon > 0$, we have
  \begin{equation*}
    \mathbb P(l + \epsilon \leq W \leq u - \epsilon) \geq \mathbb P (l \leq W \leq u) - 2\mathbb P(-\epsilon/2 \leq W \leq \epsilon/2).
  \end{equation*}
\end{lemma}
\begin{proof}
  We can write the probability on the r.h.s. as 
  \begin{equation*}
    \mathbb P (l \leq W \leq u) - \mathbb P (u-\epsilon \leq W \leq u) - \mathbb P (l \leq W \leq l + \epsilon).
  \end{equation*}
  Because the p.d.f. of a centered Gaussian distribution is symmetric around $0$ and is strictly decreasing on the positive reals, it follows that the second and third probability are bounded from above by $\mathbb P (-\epsilon/2 \leq W \leq \epsilon/2)$.
\end{proof}

We continue now with the proof of Theorem \ref{th_te2}.

\begin{proof}
  Fix an $\epsilon > 0$. Let $K \subset \mathbb R^d$ be compact and set
  \begin{equation*}
      B_1 = \{X \in K\}
  \end{equation*}
  and
  \begin{equation*}
      B_{2,n} = \left \{ \sup_{(x,t) \in K \times \{-1,1\}} |f_{D_n}(x,t) - f(x,t)| < \epsilon \right \}.
  \end{equation*}
  Since $X$ is stochastically bounded (it is $\mathbb R^d$-valued), we can choose $K$ large enough such that $\mathbb P(B_1)$ is arbitrarily close to 1. Because $\mathcal A_n$ is a consistent prediction procedure, we can choose $n$ large enough such that $\mathbb P(B_{2,n})$ is arbitrarily close to 1. This means, we can choose $K$ and $n$ large enough such that $\mathbb P(B_1 \cap B_{2,n})$ is arbitrarily close to $1$. Note that proving $\mathbb P(\tau(X) \in \Gamma_{D_n} (X)) \geq 1-\alpha$ as $n \to \infty$ is equivalent to proving that
  \begin{equation*}
    \mathbb P \left(\tau(X) \in \Gamma_{D_n} (X) ~ | ~ B_1 \cap B_{2,n} \right)
  \end{equation*}
  is bounded from below by $1-\alpha$ as $\mathbb P(B_1 \cap B_{2,n})$ converges to 1. Observe that we can write the conditional probability as
  \begin{equation*}
    \frac{\mathbb E \left[ \mathbb E [\mathds{1}_{\{\tau(X) \in \Gamma_{D_n} (X) \}} \mathds{1}_{B_1 \cap B_{2,n}} ~ | ~ X, D_n] \right]}{\mathbb P (B_1 \cap B_{2,n})}.
  \end{equation*}
  Because the event $B_1 \cap B_{2,n}$ only depends on $X$ and $D_n$ (and is therefore measurable with respect to the $\sigma$-Algebra generated by $X$ and $D_n$), the conditional expectation in the numerator is equal to $\mathbb P(\tau(X) \in \Gamma_{D_n} (X) | X, D_n)\mathds{1}_{B_1 \cap B_{2,n}}$. But this implies that the conditional probability is equal to
  \begin{equation} \label{eq_p1}
      \mathbb E \left[ \mathbb P(\tau(X) \in \Gamma_{D_n} (X) ~ | ~ X, D_n) | B_1 \cap B_{2,n} \right].
  \end{equation}
  Observe that
  \begin{align*}
    \{\tau(X) \in \Gamma_{D_n} (X)\} ~ = ~ \bigg \{ & (\tilde l_{D_n}(X,1) - f(X,1))  - (\tilde u_{D_n}(X,-1) - f(X,-1)) \\
    & \leq ~ e_{X,1} - e_{X,-1} ~ & \\
    & \leq  ~ (\tilde u_{D_n}(X,1) - f(X,1))  - (\tilde l_{D_n}(X,-1) - f(X,-1)) \bigg \}.    
  \end{align*}
  On the event $B_1 \cap B_{2,n}$, we have $|f(X,t)-f_{D_n}(X,t)| < \epsilon$ for all $t \in \{-1,1\}$. This means that (\ref{eq_p1}) becomes only smaller if we replace $\{\tau(X) \in \Gamma_{D_n} (X)\}$ by
  \begin{align*}
    \bigg \{ & (\tilde l_{D_n}(X,1) - f_{D_n}(X,1))  - (\tilde u_{D_n}(X,-1) - f_{D_n}(X,-1)) + 2\epsilon  \\
    & \leq ~ e_{X,1} - e_{X,-1} ~ & \\
    & \leq  ~ (\tilde u_{D_n}(X,1) - f_{D_n}(X,1))  - (\tilde l_{D_n}(X,-1) - f_{D_n}(X,-1)) - 2\epsilon \bigg \}.
  \end{align*}
  By definition of $\tilde l_{D_n}(X,t)$ and $\tilde u_{D_n}(X,t)$, $t \in \{-1,1\}$, the previous set is equal to
  \begin{align*}
    \bigg \{ & - \frac{f_{D_n}(X,1) - l_{D_n}(X,1))}{\sqrt{2}}  - \frac{u_{D_n}(X,-1) - f_{D_n}(X,-1)}{\sqrt{2}} + 2\epsilon  \\
    & \leq ~ e_{X,1} - e_{X,-1} ~ & \\
    & \leq  ~ \frac{u_{D_n}(X,1) - f_{D_n}(X,1)}{\sqrt{2}} + \frac{f_{D_n}(X,-1)-l_{D_n}(X,-1)}{\sqrt{2}} - 2\epsilon \bigg \}.
  \end{align*}
    Observe that conditional on $X$ and $D_n$, the only random quantities remaining are $e_{X,-1}$ and $e_{X,1}$. For $t \in \{-1,1\}$, we have by assumption that $e_{X,t}$ given $X$ is a centered Gaussian. Because $D_n$ is independent of $e_{X,t}$, the same follows for $e_{X,t}$ given $X$ and $D_n$. Therefore, Lemma \ref{le_a2} implies that (\ref{eq_p1}) becomes only smaller if we replace the conditional probability in (\ref{eq_p1}) by
  \begin{align*}
    \mathbb P \bigg( & -\frac{f_{D_n}(X,1) - l_{D_n}(X,1))}{\sqrt{2}}  - \frac{u_{D_n}(X,-1) - f_{D_n}(X,-1)}{\sqrt{2}}\\
    &\leq ~ e_{X,1} - e_{X,-1} \\
    &\leq ~ \frac{u_{D_n}(X,1) - f_{D_n}(X,1)}{\sqrt{2}} + \frac{f_{D_n}(X,-1)-l_{D_n}(X,-1)}{\sqrt{2}} ~ | ~ X,D_{n} \bigg) \\
    &+ 2\mathbb P \left(-\epsilon \leq e_{X,1} - e_{X,-1} \leq \epsilon ~ | ~ X,D_{n} \right).
  \end{align*}
  In view of (\ref{cover}), we have $f_{D_n}(X,t) - l_{D_n}(X,t) \geq 0$ a.s. and $u_{D_n}(X,t) - f_{D_n}(X,t) \geq 0$ a.s. for all $t \in \{-1,1\}$. By assumption, $\epsilon_{X,-1}$ and $\epsilon_{X,1}$, conditional on $X$, are centered multivariate Gaussian with the same variance and non-negative correlation. Because $D_n$ is independent of $e_{X,t}$, the same follows for $\epsilon_{X,-1}$ and $\epsilon_{X,1}$ given $X$ and $D_n$. Lemma \ref{le_a1} implies that the first summand in the previous expression is bounded from below by
  \begin{equation*}
    \frac{1}{2} \sum_{t \in \{-1,1\}} \mathbb P \left(l_{D_n}(X,t) - f_{D_n}(X,t)  \leq e_{X,t} \leq u_{D_n}(X,t) - f_{D_n}(X,t) ~ | ~ X,D_n \right).
  \end{equation*}
  Because $|f(X,t)- f_{D_n}(X,t)| < \epsilon$ on the event $B_{1} \cap B_{2,n}$, conditional on $B_{1} \cap B_{2,n}$, the previous sum is bounded from below by
  \begin{align*}
    &\frac{1}{2} \sum_{t \in \{-1,1\}} \mathbb P \left(l_{D_n}(X,t) - f(X,t)  \leq e_{X,t} \leq u_{D_n}(X,t) - f(X,t) ~ | ~ X,D_n \right) \\
    &- ~ \sum_{t \in \{-1,1\}} \mathbb P \left(-\epsilon/2 \leq e_{X,t} \leq \epsilon/2 ~ | ~ X,D_n \right).
  \end{align*}    
  (We replace the first $f_{D_n}(X,t)$ by $f(X,t)-\epsilon$ and the second by $f(X,t)+\epsilon$ and then apply Lemma \ref{le_a2}.) Using the same argumentation as in the paragraph above (\ref{eq_p1}) (only in the reverse direction), we can conclude that (\ref{eq_p1}) is bounded from below by
  \begin{align*}
    &\frac{1}{2} \sum_{t \in \{-1,1\}}  \mathbb P \left(l_{D_n}(X,t) - f(X,t)  \leq e_{X,t} \leq u_{D_n}(X,t) - f(X,t) ~ | ~ B_1 \cap B_{2,n} \right) \\
    &- ~ \sum_{t \in \{-1,1\}} \mathbb P \left(-\epsilon/2 \leq e_{X,t} \leq \epsilon/2 ~ | ~ B_1 \cap B_{2,n} \right) \\
    &- 2 \mathbb P \left(-\epsilon \leq e_{X,1} - e_{X,-1} \leq \epsilon ~ | ~ B_1 \cap B_{2,n} \right).
  \end{align*}
  Since we want to show that this is bounded from below by $1-\alpha$ as $\mathbb P (B_1 \cap B_{2,n})$ converges to $1$, it is equivalent to showing that
    \begin{align*}
    &\frac{1}{2} \sum_{t \in \{-1,1\}}  \mathbb P \left(l_{D_n}(X,t) - f(X,t)  \leq e_{X,t} \leq u_{D_n}(X,t) - f(X,t)  \right) \\
    &- ~ \sum_{t \in \{-1,1\}} \mathbb P \left(-\epsilon/2 \leq e_{X,t} \leq \epsilon/2 \right) \\
    &- 2 \mathbb P \left(-\epsilon \leq e_{X,1} - e_{X,-1} \leq \epsilon \right).
  \end{align*}
  is bounded from below by $1-\alpha$ as $n \to \infty$. Observe that the first summand is equal to
  \begin{equation*}
      \frac{1}{2} \sum_{t \in \{-1,1\}}  \mathbb P \left(Y \in [l_{D_n}(X,t),u_{D_n}(X,t)] ~ | ~ T=t \right)
  \end{equation*}
  which is bounded from below by $1-\alpha$ for every $n$ by assumption. Since $(e_{X,-1},e_{X,1})'$ conditional on $X$ is multivariate Gaussian, it follows that $(e_{X,-1},e_{X,1})'$ has no point mass unconditionally. Because $\epsilon$ was arbitrary, we can make the second and third summand arbitrarily small. 
  
\end{proof}

\subsection{Proofs of Theorem \ref{th_full} and \ref{th_split}}

\begin{proof}[Proof of Theorem \ref{th_full}]
    Set 
  \begin{equation*}
    N_t = \sum_{i = 1}^n \mathds{1}_{\{T_i = t\}}.
  \end{equation*}
  Observe that $N_t$ is $\mathrm{Binom}(n,p_t)$-distributed. Conditional on $\{N_t, T=t\}$, the set $D_{n+1,T}(Y)$ has cardinality $N_t+1$ and the conditional joint distribution of the random variables in $m_{D_{n+1}(Y)}(D_{n+1,T}(Y))$ is exchangeable. This implies that the conditional distribution of $R_{D_{n+1}(Y)}(Y)$ is a discrete uniform distribution on $\{1,\dots,N_t+1\}$. Hence, it follows that
  \begin{align*}
    1-\alpha ~ &\leq ~ \frac{\lceil (1-\alpha) (N_t+1) \rceil}{N_t + 1} \\
    &= ~ \mathbb P \left(R_{D_{n+1}(Y)}(Y) \leq \lceil (1-\alpha) (N_t+1) \rceil ~ | ~ N_t,T=t \right) \\
    &\leq ~ 1-\alpha + \frac{1}{N_t+1}
  \end{align*}
  % Because, conditional on $\{T^{(n)}=t^{(n)},T=t\}$,
  % \begin{equation*}
  %   n(t^{(n)},t) = \sum_{i=1}^n  \mathds{1}_{\{T_i=T\}} + 1 \text{ a.s.}
  % \end{equation*}
  % it follows by the definition of $\Gamma_{D_n}^F(X,T)$ in (\ref{eq_pi_full}) that 
  % \begin{equation*}
  %   \mathbb P \left(Y \in \Gamma_{D_n}^F(X,T) ~ | ~ T^{(n)}=t^{(n)},T=t \right) ~ \geq ~ 1-\alpha.
  % \end{equation*}
  This means that the conditional probability
  \begin{equation*}
    \mathbb P \left(R_{D_{n+1}(Y)}(Y) \leq \lceil (1-\alpha) N_t \rceil ~ | ~ T=t \right)
  \end{equation*}
  is also bounded from below by $1-\alpha$. For the upper bound, we need to compute the first moment of $(N_t+1)^{-1}$, where $N_t$ is $\mathrm{Binom}(n,p_t)$-distributed. \cite{chao1972} showed that this is equal to $(1-(1-p_t)^{n+1}) / ((n+1)p_t)$.
\end{proof}

\begin{proof}[Proof of Theorem \ref{th_split}]
    Set 
  \begin{equation*}
    \tilde N_t = \sum_{i = m+1}^n \mathds{1}_{\{T_i = t\}}.
  \end{equation*}
  Observe that $\tilde N_t$ is $\mathrm{Binom}(n-m,p_t)$-distributed. Conditional on $\{\tilde N_t, T=t\}$, the set $D_{m:(n+1),T}(Y)$ has cardinality $\tilde N_t+1$ and the random variables in $m_{D_m}(D_{m:(n+1),T}(Y))$ are exchangeable because the dataset $D_{m}$ is independent of the random variables in $D_{m:(n+1),T}(Y))$. This implies that the conditional distribution of $R_{D_m}(Y)$ is a discrete uniform on $\{1,\dots, \tilde N_t+1\}$ and therefore
  \begin{align*}
    1-\alpha ~ &\leq ~ \frac{\lceil (1-\alpha) (\tilde N_t+1) \rceil}{\tilde N_t + 1} \\
    &= ~ \mathbb P \left(R_{D_m}(Y) \leq \lceil (1-\alpha) (\tilde N_t+1) \rceil ~ | ~ \tilde N_t,T=t \right) \\
    &\leq 1-\alpha + \frac{1}{(\tilde N_t+1)}.
  \end{align*} 
  The claim now follows by the same argument as in the proof of Theorem \ref{th_full}.
\end{proof}

\end{appendices}

\bibliographystyle{apalike}
\bibliography{literatur}

\end{document}